\documentclass[12pt]{article} \usepackage{fullpage}
\usepackage{amssymb,amsmath}
\usepackage{lineno}

\usepackage{nicefrac}
\usepackage{url}
\usepackage{color}
\usepackage[american]{babel}
\usepackage{graphicx}
\usepackage{version}
\usepackage{subfigure}
\usepackage[textsize=footnotesize]{todonotes}

\newif\iffull
\fulltrue

\graphicspath{{figures/}}

\usepackage{amsthm}
\newtheorem{theorem}{Theorem}
\newtheorem{corollary}{Corollary}
\newtheorem{conjecture}{Conjecture}

\newtheorem{lemma}{Lemma}

\newcommand{\calT}{{\ensuremath{\cal T}}}

\newcommand{\leaveout}[1]{}

\date{}
\title{Are highly connected 1-planar graphs Hamiltonian?}
\author{Therese Biedl}

\begin{document}

\maketitle
\begin{abstract}
It is well-known that every planar 4-connected graph has a 
Hamiltonian cycle.  In this paper, we study the question
whether every 1-planar 4-connected graph has a  Hamiltonian
cycle.  We show that this is false in general, even for
5-connected graphs, but true if
the graph has a 1-planar drawing where every region is a triangle.
\end{abstract}


\section{Introduction}

\iffull
Planar graphs are graphs that can be drawn without crossings.
They have been one of the central
areas of study in graph theory and graph algorithms, and there
are numerous results both for how to solve problems more easily
on planar graphs and how to draw planar graphs (see e.g.~\cite{Aig84,NC88,NR04})).
We are here interested in a theorem by Tutte \cite{Tutte56} that states 
that every planar 4-connected graph has a Hamiltonian cycle 
(definitions are in the next section).  This was an improvement over
an earlier result by Whitney that proved the existence of a Hamiltonian
cycle in a 4-connected {\em triangulated} planar graph.
There have been many
generalizations and improvements since; in particular we can
additionally fix the endpoints and one edge that the Hamiltonian
cycle must use \cite{Thomassen83,Sanders97}.   Also, 4-connected
planar graphs remain Hamiltonian even after deleting 2 vertices
\cite{TY94}.
Hamiltonian cycles
in planar graphs can be computed in linear time; this is quite
straightforward if the graph is triangulated \cite{AKS84} and a
bit more involved for general 4-connected planar graphs \cite{CN89}.

There are many graphs that
are {\em near-planar}, i.e., that are ``close'' to planar graphs.
One such generalization are the {\em 1-planar graphs}, which are
the graphs that can be drawn with at most one crossing per edge
\cite{Ringel1965}.
Naturally one wonders which of the properties of planar graphs
carry over to 1-planar graphs.  Many results have been developed,
for example it is well-known that 1-planar graphs have at most
$4n-8$ edges \cite{Schumacher} and are 6-colourable \cite{Ringel1965}
and it has been characterized when 1-planar graphs have a 
straight-line drawing \cite{Tho88}.
See \cite{KLM17} for a recent overview of many existing results
for 1-planar graphs.

This paper investigates results on Hamiltonicity
in sufficiently connected 1-planar graphs.  In particular, we show that
every 4-connected {\em triangulated} 1-planar graph has a Hamiltonian cycle.
This is done via a detour: we show that (with the exception of $K_5$)
every triangulated
1-planar 4-connected graph contains a triangulated planar
4-connected graph as a subgraph, and then appeal to Tutte's
theorem.  This in particular implies that all the generalizations,
such as fixing the endpoints or one edge to be visited, or staying
Hamiltonian after deleting two vertices, carry
over to triangulated 1-planar 4-connected graphs.  

The argument
crucially requires that the graph is triangulated:  we can easily construct
a 4-connected 1-planar graph that does not have a Hamiltonian path (and
not even a near-perfect matching, which is a weaker condition).  In fact,
even 5-connected 1-planar graphs do not always have near-perfect matchings,
while the question remains open for 1-planar graphs of higher connectivity.

\else
Write a little bit here
\fi

\section{Background}

Let $G=(V,E)$ be a graph with $n$ vertices and $m$ edges.
We assume familiarity with graph theory (see e.g.~\cite{Die12}) and
review only some of the notations below.   All graphs in this paper are
{\em simple}, i.e., no edge connects a vertex with itself and
no two edges connect the same pair of vertices.  A {\em subgraph} of $G$
is obtained by deleting some vertices (and all their incident edges)
and/or deleting some edges.  A {\em spanning subgraph} is obtained by
deleting only edges.  A graph $G$ is called {\em connected}
if any two vertices $v,w$ are connected by a path within $G$.
A {\em connected
component} of $G$ is a maximal subgraph that is connected.  A {\em cutting
$k$-set} is a set $S$ of at most $k$ vertices such that $G\setminus S$ has
more connected components than $G$.  Graph $G$ is called {\em $k$-connected}
if it has no cutting $(k{-}1)$-set, and if further it has at least $k$
vertices.  Menger's theorem states that $G$ is $k$-connected if and only
if for any two vertices $v,w$ there are at least $k$ paths from $v$ to $w$
that are {\em interior vertex-disjoint}, i.e., they have no vertex in common except the
ends $v$ and $w$.

A {\em Hamiltonian path/cycle} is a path/cycle in the graph that visits
every vertex exactly once.  A {\em matching} is a set $M$ of edges such that
no two edges in $M$ have a common endpoint.  It is called {\em near-perfect}
if $M=\lfloor n/2 \rfloor$, which is to say, it is as big as any matching
can be in an $n$-vertex graph.  Note that any graph with a Hamiltonian path
also has a near-perfect matching, though the reverse is not always true.

A graph $G$ is called {\em planar} if it can be drawn in the plane without
crossing.  If one particular such drawing $\Gamma$
is fixed then $G$ is called {\em plane}. 
The maximal connected regions of $\mathbb{R}\setminus \Gamma$ are called {\em faces}.
Graph $G$ is called {\em triangulated} if every face is bounded by a triangle.  A
triangulated simple planar graph is 3-connected and has a unique planar
embedding.

A graph is called {\em 1-planar} if it can be drawn in the plane
such that every edge has at most one crossing.  We assume here that the
drawing is {\em good}, which means that no edges with a common endpoint
cross and no edge crosses itself.   If one particular
such drawing $\Gamma$ is fixed then $G$
is called {\em 1-plane}.  In a 1-plane graph an 
edge is called {\em crossed} if it
has a crossing and {\em uncrossed} otherwise.  Generalizing, a cycle 
is called {\em crossed} if at least one of its edges is crossed and
{\em uncrossed} otherwise.
The maximal connected regions of $\mathbb{R}\setminus \Gamma$ are bounded
by a sequence of vertices and crossings; these are called the {\em corners}
of the region.  
A {\em triangulated} 1-plane graph is one where
all regions have exactly three corners.    Such a graph may have multiple
1-planar embeddings, and not all of them are necessarily triangulated
(for example $K_4$ has a triangulated drawing without crossing, or a
drawing with one crossing that is not triangulated.)

Let $G$ be a graph with a fixed 
(planar or 1-planar) drawing $\Gamma$.  A {\em separating cycle} of $G$ is
a cycle $C$ such that the curve traced by $C$ in $\Gamma$ does not self-intersect and has at least one vertex strictly inside and strictly outside.
It is easy to see that if the drawing is planar and triangulated, then
the graph is $k$-connected if and only if it has no separating
cycle of length at most $k{-}1$.  To our knowledge no equivalent
characterization is known for 1-planar drawing.  Obviously, if a 1-planar
drawing has a {\em separating uncrossed triangle} $T$, i.e., a separating
3-cycle where no edge is crossed, then the vertices of $T$ form a cutting
3-set and the graph is not 4-connected.
As part of our exploration it follows (see
Corollary~\ref{cor:4conn1planar}) that this exactly characterizes when
a 1-planar triangulated drawing represents a 4-connected graph. 

\section{Making 1-planar graphs planar}

It is quite obvious that any triangulated 1-plane graph can be made 
planar by deleting one of each pair of crossing edges.  Furthermore,
the resulting graph is triangulated.  We argue here that if we are
more carefully about which of the crossing edges is removed, then we
can additionally ensure that no uncrossed separating triangle is
created.

\begin{lemma}
\label{lem:removeEdge}
Let $G$ be a triangulated 1-plane graph without uncrossed separating triangle,
and assume that $G$ contains at least 6 vertices.
Then $G$ contains a spanning subgraph $G^-$ that is planar, 
triangulated, and has no separating triangle.
\end{lemma}
\begin{proof}
We proceed by induction on the number of crossings in $G$.  If there
is none, then $G$ itself is planar and triangulated.  Any separating
triangle $T$ of $G$ would be uncrossed by planarity, so none can exist
by assumption.

Now assume that $G$ has a crossing, say $(v_1,v_2)$ crosses $(w_1,w_2)$.  
Since $G$ is triangulated, the 4-cycle $\langle v_1,w_1,v_2,w_3\rangle$
also exists and consists of uncrossed edges.  See Figure~\ref{fig:removeEdge}.

\medskip\noindent{\bf Case 1:}
Assume first that for any vertex $v_3\neq v_1,v_2,w_1,w_2$, either $\{v_1,v_2,v_3\}$
	is not a triangle, or at least one of the edges $(v_1,v_3)$ and $(v_2,v_3)$ is crossed.  See Figure~\ref{fig:removeEdge}(a).

	In this case, delete edge $(w_1,w_2)$.  We claim that the resulting
	graph 1-plane $G'$ satisfies the conditions of the lemma, i.e., it is
	triangulated 1-plane and without uncrossed separating triangles.  
	The result then follows by induction since the spanning
	subgraph of $G'$ is also a spanning subgraph of $G$.

	Obviously $G'$ is 1-plane in
	the inherited embedding.  
	Deleting $(w_1,w_2)$ results in two regions 
	$\{v_1,v_2,w_1\}$ and $\{v_1,v_2,w_2\}$, while all other regions
	are unchanged, so $G'$ is again triangulated.
	Any
	separating triangle $T$ of $G'$ is also a separating triangle of $G$.  
	If $T$ were uncrossed in $G'$ but crossed in $G$, then $T$ must include
	edge $(v_1,v_2)$. So $T=\{v_1,v_2,v_3\}$ for some $v_3\neq v_1,v_2$,
	and also $v_3\neq w_1,w_2$ because $T$ is separating while 
	$\{v_1,v_2,w_1\}$ and
	$\{v_1,v_2,w_2\}$ bound regions in $G'$.
	But by case assumption one of $(v_1,v_3)$ and $(v_2,v_3)$ is crossed,
	contradicting that $T$ is uncrossed.

\medskip\noindent{\bf Case 2:}
Now assume that there exists a vertex $v_3\neq v_1,v_2,w_1,w_2$ for which
	edges $(v_1,v_3)$ and $(v_2,v_3)$ exist and are uncrossed in $G$.
	See Figure~\ref{fig:removeEdge}(b).

	Delete edge $(v_1,v_2)$ from the drawing and call the resulting
	1-plane graph $G'$.  We claim
	that $G'$ satisfies the conditions of the lemma; the result then follows
	by induction.  As before $G'$ is triangulated 1-plane.  Assume for
	contradiction that it contains a separating uncrossed triangle, which necessarily has the
	form $\{w_1,w_2,w_3\}$ for some vertex $w_3$, since $(w_1,w_2)$ is
	the only edge that was crossed in $G$ but uncrossed in $G'$.

	The drawing of $G$  contained triangle $T_v=\{v_1,v_2,v_3\}$,
	which forms a closed curve.  Edge $(w_1,w_2)$ intersects $T_v$ once
	at $(v_1,v_2)$ and cannot intersect it again by 1-planarity, hence
	(up to renaming) $w_1$ is inside $T_v$ while $w_2$ is outside.  But
	there is a path $\pi=\langle w_1,w_3,w_2\rangle$ from $w_1$ to $w_2$ in $G$,
	which also must cross $T_v$.  Since edge $(v_1,v_2)$ cannot be crossed
	again in a 1-planar graph, and edges $(v_1,v_3)$ and $(v_2,v_3)$ are
	uncrossed by case assumption, this implies that some vertex of $\pi$
	must be on $T_v$.  But none of $\{v_1,v_2,w_1,w_2\}$ coincide since they
	participate in a crossing.  We also know $w_3\neq v_1,v_2$ since
	$\{w_1,w_2,v_1\}$ and $\{w_1,w_2,v_2\}$ are regions of $G'$, not
	separating triangles.  Therefore
	we must have $w_3=v_3$.  

	But then $w_3=v_3$ is adjacent to all of $\{v_1,v_2,w_1,w_2\}$.
	Since $\{v_1,v_2,w_1,w_2\}$ form a $K_4$, this gives a $K_5$.
	See also Figure~\ref{fig:removeEdge}(c).
	Furthermore, all edges incident to $v_3=w_3$ in this $K_5$ are
	uncrossed in $G$, by case assumption and since $\{w_1,w_2,w_3\}$ 
	is an uncrossed triangle in $G'$.  
	We claim that this is impossible. 
	Namely, by $n\geq 6$ at least one region $R$ of this 
	$K_5$ contains additional vertices.  If $R$ is incident to $w_3$,
	then its boundary is an uncrossed triangle; this would make $R$
	an uncrossed separating triangle, a contradiction.  If $R$ is
	not incident to $w_3$, then it is incident to the crossing $c$
	of $(v_1,v_2)$ and $(w_1,w_2)$, and an edge $(a,b)$ of the
	uncrossed cycle $\langle v_1,w_1,v_2,w_3\rangle$.
	Since $(a,b)$ is uncrossed,
	and the part-edges from $a,b$ to $c$ cannot be crossed again,
	the vertices inside $R$
	can be adjacent only to $a,b$, making $\{a,b\}$ a cutting 2-set
	of $G$.  But one easily convinces oneself that a 
	triangulated simple 1-planar graph is 3-connected (for example because
	we can delete crossed edges until we obtain a triangulated planar graph
	as a spanning subgraph), so this is a contradiction.
\end{proof}

\begin{figure}[t]
\hspace*{\fill}
\subfigure[~]{\includegraphics[scale=1,page=1,trim=0 0 0 0,clip]{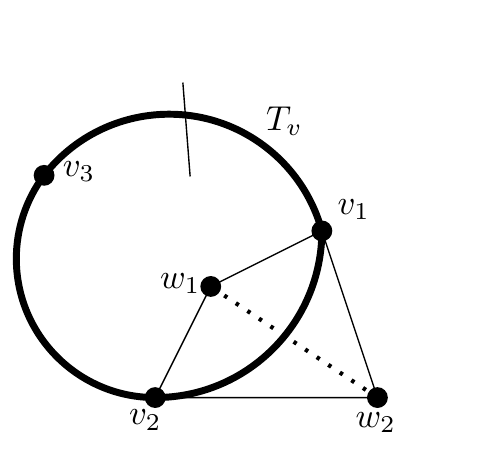}}
\hspace*{\fill}
\subfigure[~]{\includegraphics[scale=1,page=2,trim=0 0 0 0,clip]{tight.pdf}}
\hspace*{\fill}
\subfigure[~]{\includegraphics[scale=1,page=4,trim=0 0 0 0,clip]{tight.pdf}}
\hspace*{\fill}
\caption{
We can remove one of the crossing edges $(v_1,v_2)$ and $(w_1,w_2)$ without
creating a separating uncrossed triangle.  (a) Any triangle $T_v$ containing
$(v_1,v_2)$ has at least one other crossed edge.  (b) Some triangle $T_v$ 
containing $(v_1,v_2)$ has no other crossed edges. (c) $G$ cannot contain
a $K_5$ where all edges to $v_3$ are uncrossed.
}
\label{fig:removeEdge}
\end{figure}

This implies a few useful results.

\begin{corollary}
\label{cor:4conn1planar}
\label{cor:1planar4conn}
Let $G$ be a 1-plane triangulated graph with $n\geq 6$.  Then $G$ is 4-connected
if and only if $G$ contains no uncrossed separating triangle.
\end{corollary}
\begin{proof}
If $G$ contains an uncrossed separating triangle $T$, then the vertices
of $T$ form a separating 3-set since none of the edges of $T$ are crossed.
Therefore $G$ is not 4-connected.

Now assume that $G$ contains no uncrossed separating triangle.  
By Lemma~\ref{lem:removeEdge}, we can find a subgraph  $G^-$ that
is planar and triangulated and has no separating triangle.  Therefore
$G^-$ (and with it its supergraph $G$) is 4-connected.
\end{proof}

\begin{theorem}
\label{thm:main}
Any 4-connected triangulated 1-plane graph $G$ has a Hamiltonian cycle.
\end{theorem}
\begin{proof}
The claim clearly holds for $n\leq 5$, because there are only three
4-connected triangulated 1-plane graphs with $n\leq 5$  ($K_4$, $K_5$ and
$K_5\setminus e$) and all three have Hamiltonian cycles.  For $n\geq 6$,
a 4-connected triangulated 1-plane graph $G$ has no separating
uncrossed triangle, so by Lemma~\ref{lem:removeEdge} we can find a subgraph $G^-$ 
that is planar, triangulated and has no separating triangle.  Graph $G^-$ is 4-connected
and planar, so we can find a Hamiltonian cycle in $G^-$ (hence also $G$).
\end{proof}

A 1-planar graph is called an {\em optimal 1-planar graph} if it
has $4n-8$ edges (the maximum possible number in a 1-planar graph).
It is known that every optimal 1-planar graph is 4-connected and
triangulated \cite{Schumacher} and hence we have:

\begin{corollary}
Every optimal 1-planar graph has a Hamiltonian cycle.
\end{corollary}

\subsection{Linear run-time}

Finding a Hamiltonian cycle in a planar triangulated graph can be
done in linear time \cite{AKS84}.  To find the Hamiltonian cycle 
of Theorem~\ref{thm:main} likewise in linear time,
we therefore must argue that the
proof of Lemma~\ref{lem:removeEdge} gives rise to an algorithm
that runs in linear time.

Fix a 1-planar graph $G$.  As a first step, we compute all triangles
in $G$, using the algorithm by Chiba and Nishizeki \cite{Chiba85b}.
To understand its run-time, we need a minor detour.
Define the {\em arboricity} $a(G)$ to be the smallest number
$k$ such that the edges of $G$ can be partitioned into $k$ sets, each of which
forms a forest.  By the Nash-Williams formula \cite{NW61} we know that
$a(G)=\max_H \left\lceil |E(H)| / (|V(H)|-1) \right \rceil$, where the maximum goes over all subgraphs $H$
of $G$ with at least $2$ vertices.  It is known that an $n$-vertex 1-planar
simple graph has at most $4n-8$ edges, except for $n=2$ where it has at most $4n-7$ edges.
Therefore the arboricity of $G$ is at most 4.  The triangle-finding algorithm by Chiba
and Nishizeki has run-time $O(a(G) n)$, which therefore is $O(n)$.  (This in particular
implies that there is only a linear number of triangles.)

We store a list $\calT$ of triangles, and cross-link all triangles to the edges that
contain them; this takes time $O(|\calT|)=O(n)$.  Within the same time we can also
build for each edge $e$ a list $\calT_e$ of triangles that contain $e$. 
We also assume that from the 1-planar embedding
we have a list of all crossings that are cross-linked to the two crossing edges. 
Now we parse all crossings.  For each crossing $c$,
look up the edges $e,e'$ involved in it.  After possible renaming, $|\calT_e|\leq |\calT_{e'}|$.
For each triangle $T \in \calT_e$, check whether any of the other edges of $T$ are crossed.
If not, then check whether $T$ consists of $e$ and one endpoint of $e'$.
If not, then we have found a separating uncrossed triangle and we delete $e'$ and are done parsing $c$.
Otherwise (if none of the triangles in $\calT_e$ are separating uncrossed) we
delete $e$.  The run-time to determine this choice is $O(1+|\calT_e|)$.

Deletion of one edge $\hat{e}\in \{e,e'\}$ consists of updating that the other edge 
is now uncrossed, and deleting all triangles in $\calT_{\hat{e}}$.
This takes time $O(1+|\calT_{\hat{e}}|)$.  Since $|\calT_e|\leq |\calT_{\hat{e}}|$,
the entire run-time for removing one
crossing is proportional to the number of deleted triangles plus one. This is overall
linear time since there are $O(n)$ triangles and $O(n)$ crossings.
We conclude:

\begin{theorem}
Given a 4-connected triangulated 1-planar graph $G$, we can find a
4-connected triangulated planar spanning subgraph of $G$ in linear time.
\end{theorem}

\begin{corollary}
Every 4-connected triangulated 1-planar graph $G$ has a Hamiltonian
cycle, and it can be found in linear time.
\end{corollary}

\section{Non-hamiltonian 1-planar graphs}

In this section, we exhibit some 1-planar graphs that
do not have a Hamiltonian path, and in fact, do not even have a
near-perfect matching.  

\subsection{Could Theorem~\ref{thm:main} be made stronger?}

Theorem~\ref{thm:main} requires two assumptions on the 1-planar graph: It must be
triangulated and it must be 4-connected.  The latter is obviously required; one
can easily construct planar (hence 1-planar) triangulated graphs that have no Hamiltonian path
(and indeed, no near-perfect matching).  In fact, there are even non-Hamiltonian
1-plane graphs that are {\em maximal}:
we cannot add any edge to it without destroying 1-planarity or simplicity.

\begin{lemma}
\label{lem:max1p}
For any $N$,
there exists a simple maximal 1-plane graph $G$ with $n\geq N$ vertices that
has no Hamiltonian path.
\end{lemma}
\begin{proof}
The following construction was given by Wittnebel \cite{BW19} and is shown in
Figure~\ref{fig:tight}(a);    we repeat it here for completeness.
Set $h\geq \max\{6,\frac{N+8}{5}\}$ to be an even number, and let $H$ be the $h$-vertex
graph that consists of an $(h-2)$-cycle $C$ plus two more vertices adjacent to all vertices
of $C$.  Note that $H$ has even vertex-degrees, hence its dual graph is bipartite
and the $2h-4$ faces can be coloured as $h-2$ ``black'' faces and $h-2$ ``white''
faces such that no two faces of the same colour share an edge.

To obtain $G$, we replace all white faces by a $K_4$ and all black 
faces by a $K_6$.    The resulting graph $G$ has $h+(h-2)+3(h-2)=5h-8\geq N$
vertices, and Figure~\ref{fig:tight}(a) shows that it is 1-planar.
One easily verifies that it is maximal:  We could add an edge while staying
1-planar only within two adjacent uncrossed regions. This exists only at the $K_4$'s, 
and the edges we could add there would be double edges.

Observe that $G\setminus V(H)$ has $2h-4$ components (one per face of $H$)
that all have odd size, and we removed $h$ vertices.
By the ``easy'' part of the Tutte-Berge formula \cite{Berge1958}
any matching of $G$ has at least $h-4$ unmatched vertices.
By $h\geq 6$ hence $G$ has no near-perfect matching and no
Hamiltonian cycle.
\end{proof}

\begin{figure}[t]
\hspace*{\fill}
\subfigure[~]{\includegraphics[scale=1.2,page=3,trim=0 0 0 0,clip]{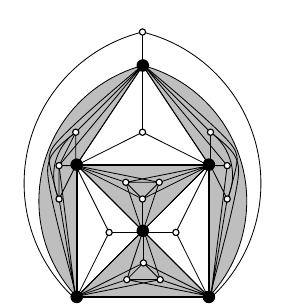}}
\hspace*{\fill}
\subfigure[~]{\includegraphics[scale=1.2,page=3,trim=0 0 0 0,clip]{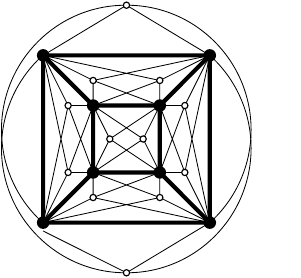}}
\hspace*{\fill}
\caption{
Two 1-planar graphs that do not have near-perfect matchings or a Hamiltonian
path.
(a) A triangulated 1-plane graph. 
(b) A 4-connected 1-planar graph. 
}
\label{fig:tight}
\end{figure}

It is less obvious that the ``triangulated'' assumption of Theorem~\ref{thm:main}
is also required.  For planar graphs, 4-connectivity alone is enough to guarantee
the existence of a Hamiltonian cycle, but the situation is different for 1-planar
graphs.

\begin{lemma}
\label{lem:mindeg4}
For any $N$,
there exists a 4-connected 1-planar graph with $n\geq N$ vertices that
has no Hamiltonian path.  In particular, any matching
has size at most $\frac{n+4}{3}$.
\end{lemma}
\begin{proof}
Consider the graph in Figure~\ref{fig:tight}(b), which has been built as follows.
Start with a simple planar graph $H$ where every face is a 4-cycle, and 
where $h:=|V(H)| \geq \max\{4,\frac{N+4}{3}\}$.
Let $H_s$ be the graph obtained by {\em stellating} every face of $H$, i.e., by inserting
into every face of $H$ a new vertex that is adjacent to all vertices of the face.
Let $G$ be the graph obtained by {\em double-stellating} every face of $H$, i.e., by inserting
into every face of $H$ {\em two} new vertices that are adjacent to all vertices of the face.

Figure~\ref{fig:tight}(b) shows that $G$ is 1-planar.
Also, $H$ has $h{-}2$ faces, hence $G$ has $n=h+2(h{-}2)=3h-4\geq N$ vertices.
Observe that $G\setminus V(H)$ has
$2(h{-}2)$ components (two per face of $H$) that are singleton vertices.  
Again using the easy part of the Tutte-Berge formula \cite{Berge1958},
any matching $M$ in $G$ leaves at least $h-4$ unmatched vertices.  So there are at
most $2h$ matched vertices and $|M|\leq h=\frac{n+4}{3}$.

It remains to argue that $G$ is 4-connected.  Observe first that $H$ is bipartite
and has no triangle, in particular it has no separating triangle.  Also, since all
its faces are 4-cyles and it is bipartite, no two non-consecutive vertices on a
face of $H$ are adjacent.  Therefore stellating $H$ does not create a
separating triangle, and it makes the graph triangulated, which means that $H_s$ is 4-connected.
To show that $G$ is
4-connected it suffices to argue that there are four interior vertex-disjoint paths from
$v$ to $w$  for any two vertices $v,w$ in $G$. 
If $v$ and $w$ have been inserted into the same face of $H$, then this
is obvious, since they have four common neighbours.  Otherwise, $v$ and $w$ existed
also as vertices in $H_s$, and we can find four paths between them
already in $H_s$.
\end{proof}

\subsection{A 5-connected non-Hamiltonian 1-planar graph}

We already exhibited a 4-connected 1-planar graph that does not
have a near-perfect matching (and hence no Hamiltonian path).  We
now show an example that shows that even 5-connectivity is not
enough (though the graph is very close to having a perfect matching).

\begin{lemma}
\label{lem:mindeg5}
For any $N$,
there exists a 5-connected 1-planar graph with $n\geq N$ vertices for which any matching has size at most $\frac{n-2}{2}$.
\end{lemma}
\begin{proof}
The graph $G$ is illustrated in Figure~\ref{fig:wall}(d);
its construction is not difficult but will be given here
via a number of subgraphs $W,W_s,W',W_s'$ because these are useful
for arguing 5-connectivity.  
Fix an integer $k\geq \max\{1,(N-2)/40\}$.
Common to all constructed graphs are $2k+1$ cycles $C_0,\dots,C_{2k}$, where $C_0$
and $C_{2k}$ have 5 vertices each while $C_1,\dots,C_{2k-1}$ have 10 vertices
each.  These cycles are arranged as nested cycles (in the figures they are drawn as horizontal lines on the flat cylinder).

\begin{figure}[ht]
\subfigure[~]{\includegraphics[width=0.49\linewidth,page=1,trim=30 0 40 0,clip]{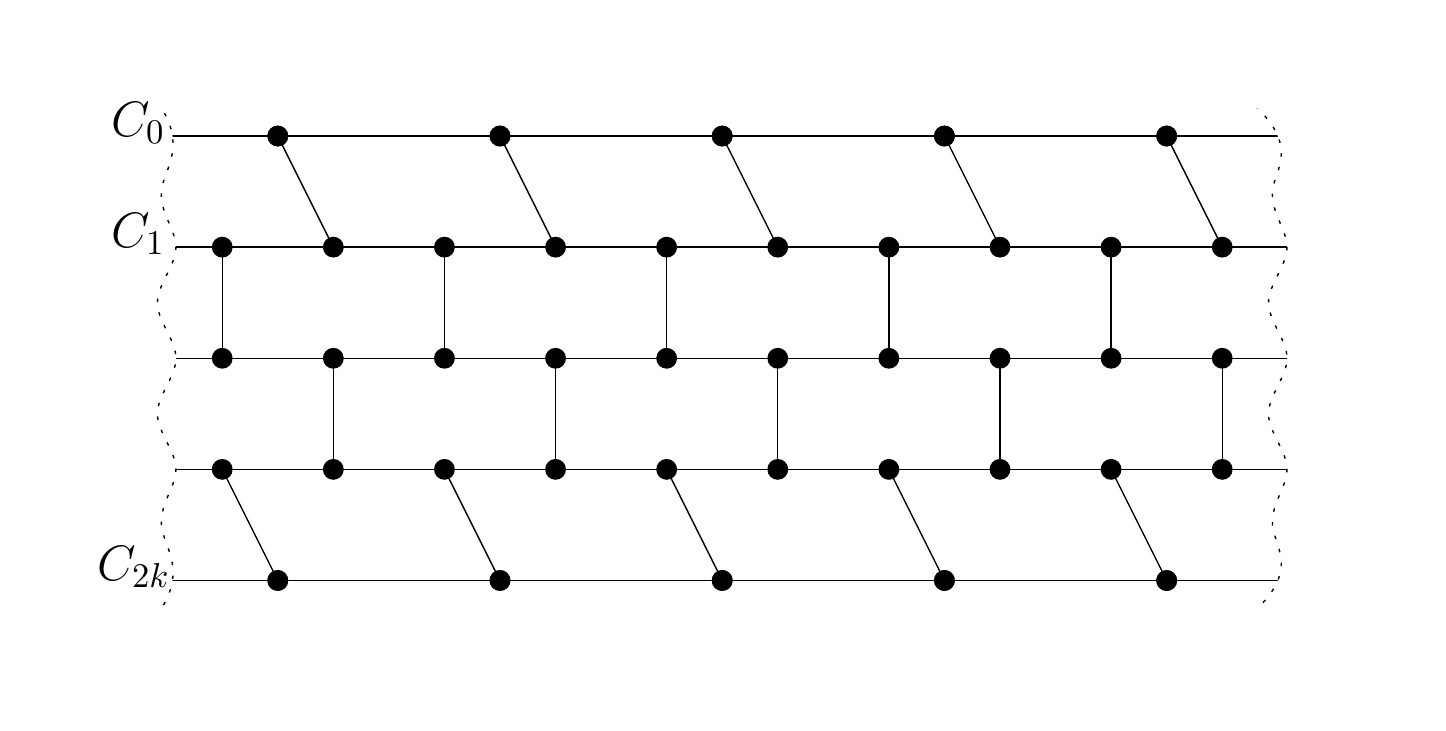}}
\hspace*{\fill}
\subfigure[~]{\includegraphics[width=0.49\linewidth,page=2,trim=30 0 40 0,clip]{wall.pdf}}
\vspace*{2mm}
\newline
\subfigure[~]{\includegraphics[width=0.49\linewidth,page=3,trim=30 0 40 0,clip]{wall.pdf}}
\hspace*{\fill}
\subfigure[~]{\includegraphics[width=0.49\linewidth,page=4,trim=30 0 40 0,clip]{wall.pdf}}
\caption{(a) The wall $W$.  (b) Its stellation $W_s$. 
(c) The stellation $W_s'$ of the other wall $W'$.   Five paths between
two vertices are in bold red.
(d) A 5-connected 1-planar graph without a near-perfect matching.
The corresponding paths are bold red.}
\label{fig:wall}
\end{figure}

Graph $W$ is the {\em wall graph}, where we add {\em connector-edges} between consecutive cycles such that all vertices have degree 3,
all faces incident to $C_0$ or $C_{2k}$ have degree 5 and all other faces have degree 6.
(For $k=1$ this is exactly the dodecahedron.)  In Figure~\ref{fig:wall}(a) the connector-edges are drawn vertically, except those incident to $C_0$ or $C_{2k}$, which are drawn downward-diagonal.  If instead we had used upward-diagonals from $C_0$ and $C_{2k}$, then we would get a graph $W'$ that is isomorphic to $W$, but uses the ``other'' vertical segments as connector-edges, see Figure~\ref{fig:wall}(c).

Let $W_s$ be the graph obtained from $W$ by stellating all faces, see Figure~\ref{fig:wall}(b).  Note that $W_s$ is triangulated and (as one verifies) has no separating 4-cycle; therefore $W_s$ is 5-connected.  The vertices in $S:=V(W_s)\setminus V(W)$ are called the {\em stellation vertices}; let $x_0,x_{2k}\in S$ be the stellation-vertices of the faces bounded by $C_0$ and $C_{2k}$.
Likewise let $W'_s$ be the stellation of $W'$, and let $S'$ be its stellation-vertices.  We use the {\em same} vertices $x_0$ and $x_{2s}$ to stellate $C_0$ and $C_{2k}$, but all other vertices in $S'$ are different from the ones in $S$.

Now define $G$ as follows.  It consists of cycles $C_0,\dots,C_{2k}$, stellation-vertices $S\cup S'$ (with $x_0$ and $x_{2k}$ added only once), and all edges incident to $S\cup S'$.  Put differently, $G$ is the union of $W_s$ and $W'_s$ after deleting the connector-edges.
Figure~\ref{fig:wall}(d) shows that $G$ is 1-planar.
Also, $W$ has $20k$ vertices and $10k{+}2$ faces, as does $W'$.  Since
$x_0$ and $x_{2k}$ are added only once, graph $G$ has hence
$n=20k+2(10k{+}2)-2=40k+2\geq N$ vertices.
Since the $20k+2$ stellation-vertices form an independent
set,  and there are only $20k$ other vertices, any matching has
at least two unmatched vertices.

It remains to argue that $G$ is 5-connected.  Roughly speaking, this
holds because $W_s$ and $W_s'$ are 5-connected, and $G$ contains 
subdivisions of $W_s$ and $W_s'$ as subgraphs (each connector-edge can be replaced
by a path through a stellation-vertex of the other wall-graph, see
Figure~\ref{fig:wall}(c-d)).  Formally, assume for
contradiction that $Q$ is a cutting-4-set, and let $y,y'$ be two vertices
in two distinct components of $G\setminus Q$.  Fix one vertex $z$ that 
belongs to one of the cycles $C_0,\dots,C_{2k}$, and furthermore $z\neq y,y'$
and $z\not \in Q$.   We claim that both $y$ and $y'$ are connected to $z$
in $G\setminus Q$; this is a contradiction since then $y,z,y'$ would all be
in one connected component of $G\setminus Q$.

We only show the existence of a path from $y$ to $z$ in $G\setminus Q$,
the argument is identical for $y'$.  We may also assume that $y\in W_s'$,
for if it is only in $W_s$ then a symmetric argument applies.  Since $z$
belongs to one of the cycles, also $z\in W_s'$.
Since $W'_s$ is 5-connected there are five interior vertex-disjoint paths
from $y$ to $z$ in $W'_s$.  Because $G$ contains a subdivision of $W'_s$
as a subgraph, these paths in $W'_s$ transfer to five interior
vertex-disjoint paths from $y$ to $z$ in $G$.
Since $y,z\not\in Q$ and $|Q|=4$, at most four of these paths can
be ``hit'' by $Q$, which means that at least one path exists even
in $G\setminus Q$, which makes $y$ connected to $z$ as desired.
\end{proof}

Note that the constructed graph is very close to having a near-perfect
matching, and we believe that this holds in general.

\begin{conjecture}
Every 5-connected 1-planar graph has a matching of size $\frac{n}{2}-O(1)$.
\end{conjecture}

\section{Conclusion}

In this paper, we studied the Hamiltonicity of 4-connected 1-planar graphs,
and showed that while in general they do not have a Hamiltonian path,
they always have a Hamiltonian cycle if they are 4-connected and triangulated.  

Among the most interesting open questions is whether higher connectivity
implies Hamiltonicity in all 1-planar graphs.  We have not been able to
construct a 6-connected 1-planar graph that does not have a near-perfect 
matching.  Do all 6-connected 1-planar graphs have a Hamiltonian cycle?
A Hamiltonian path?  Or at least a near-perfect matching?  How about
7-connected 1-planar graphs?  

\bibliographystyle{plain}
\bibliography{journal,full,gd,papers}

\end{document}